\documentclass[letterpaper, 10 pt, conference]{ieeeconf} 

\IEEEoverridecommandlockouts                             

\title{\LARGE \bf
Optimal Trajectory Control of Geometrically Exact Strings with Space-Time Finite Elements$^{*}$
}

\author{Tobias Thoma$^{1}$ and Paul Kotyczka$^{1}$
\thanks{$^{*}$Submitted to the 23rd European Control Conference (ECC).}
\thanks{$^{1}$Technical University of Munich, TUM School of Engineering
	and Design, Chair of Automatic Control, 85748 Germany (e-mail: \{tobias.thoma, kotyczka\}@tum.de)}%
}

\usepackage{graphicx}
\usepackage{amsmath}
\usepackage{xcolor}
\usepackage{amssymb}
\usepackage{color}
\usepackage{transparent}
\graphicspath{{./fig/}}

\def\ttM{\mathrm{M}}
\def\ttG{\mathrm{G}}
\def\ttC{\mathrm{C}}
\def\ttB{\mathrm{b}}
\def\ttk{\mathrm{k}}
\def\ttr{\mathrm{r}}
\def\ttu{\mathrm{u}}
\def\tty{\mathrm{y}}
\def\ttw{\mathrm{w}}
\def\ttv{\mathrm{v}}
\def\ttf{\mathrm{f}}
\def\ttx{\mathrm{x}}

\newtheorem{assum}{Assumption}
\newtheorem{rem}{Remark}
\newtheorem{prob}{Problem}
\newtheorem{lem}{Lemma}
\newtheorem{prop}{Proposition}

\begin{document}

\maketitle
\thispagestyle{empty}
\pagestyle{empty}

%%%%%%%%%%%%%%%%%%%%%%%%%%%%%%%%%%%%%%%%%%%%%%%%%%%%%%%%%%%%%%%%%%%%%%%%%%%%%%%%
\begin{abstract}
In this contribution, we present a variational space-time formulation which generates an optimal feed-forward controller for geometrically exact strings. More concretely, the optimization problem is solved with an indirect approach, and the space-time finite element method translates the problem to a set of algebraic equations. Thereby, only the positional field and the corresponding adjoint variable field are approximated by continuous shape functions, which makes the discretization of a velocity field unnecessary. In addition, the variational formulation can be solved using commercial or open source finite element packages. The entire approach can also be interpreted as a multiple-shooting method for solving the optimality conditions based on the semi-discrete problem. The performance of our approach is demonstrated by a numerical test.
\end{abstract}

%%%%%%%%%%%%%%%%%%%%%%%%%%%%%%%%%%%%%%%%%%%%%%%%%%%%%%%%%%%%%%%%%%%%%%%%%%%%%%%%
\section{Introduction}
\label{sec:1}
Increasing efficiency of mechatronic devices is often achieved through lightweight design using flexible mechanical structures like beams, strings or plates, and so on. This leads to a more complex modeling procedure often relying on the theory of nonlinear continuum mechanics. Due to the limited number of actuators the considered systems are highly under-actuated. This makes their analysis and control design challenging. A common strategy to handle such systems is a two degree of freedom control structure containing a feed-forward and a feed-back part, see \cite{Wang2021,Morlock2022} for recent examples. While a simple feed-back controller reacts to unexpected disturbances, an advanced feed-forward part is responsible for highly accurate tracking of a given reference trajectory. In this contribution, we will only focus on the feed-forward part for reference tracking.

Starting in 1995, Fliess et al. \cite{FLIESS1995} coined the concept of flatness, thus achieving a new milestone for the feed-forward control design of nonlinear systems. Flat systems offer the possibility of a purely algebraic inversion-based control design \cite{Rothfuss1997,Hagenmeyer2004}.
The concept of differential flatness has been extended for infinite-dimensional systems including several mechanical structures, e.g., \cite{Petit2001}.
Kn\"uppel and Woittennek \cite{Knueppel2010} designed a feed-forward controller for a nonlinear string based on the theory of flatness and the method of characteristics.
In the case of non-flat systems the feed-forward control design is more difficult and we have to distinguish between minimum and non-minimum phase systems \cite{Khalil2014}. 
Since these system classes contain an internal dynamics their inversion usually involves solving a two-point boundary value problem.
Chen and Paden \cite{DEGANG1996} presented a stable inversion of nonlinear non-minimum phase systems to generate a bounded solution for the system state and consequently the control input to realize the desired output trajectory with pre- and post-actuation phases.
Taylor and Li \cite{Taylor2002} introduced finite difference methods for the stable inversion of nonlinear systems.
In order to obtain a stable inversion without pre- and post-actuation phase, Graichen et al. \cite{Graichen2005} modified the boundary value problem.
Besides inversion-based feed-forward control one can also design an optimal trajectory tracking controller, see \cite{Rao2010} for a survey on numerical methods for optimal control.

Blajer and Ko\l odziejczyk \cite{Blajer2004} introduced the servo-constraint framework forming a set of differential algebraic equations for under-actuated mechanical systems. Their work had a huge influence on the multibody system dynamics community and motivated several research articles \cite{Druecker2023a}. In general, these articles focus either on flat, minimum or non-minimum phase systems.
E.g., Otto and Seifried \cite{Otto2017} applied the servo-constraint framework to a flat crane system for real-time trajectory control. 
Otto and Seifried \cite{Otto2018} demonstrated the challenges of servo-constraints for systems with a high differential index.
Str\"ohle and Betsch \cite{Stroehle2022} solved the inverse dynamics of geometrically exact strings using servo-constraints and a space-time finite element discretization.
Br\"uls et al. \cite{Bruels2013} extended the approach of Chen and Paden \cite{DEGANG1996} to the servo-constraint framework. In order to simplify the control design for non-minimum phase systems, Dr\"ucker and Seifried \cite{Druecker2023b} approximated the original boundary value problem.
Bastos et al. \cite{Bastos2013} formulated an optimization problem using servo-constraints for solving the inverse dynamics. Bastos et al. \cite{Bastos2017} also showed that their optimization problem converges to the original two-point boundary value problem of Br\"uls
et al. \cite{Bruels2013} as the pre- and post-actuation phase goes to infinity.
An analysis of alternative cost functionals is given in \cite{Altmann2016}.

To the best of our knowledge, flatness-based approaches have commonly been used for the control of strings.
These were mostly based on either infinite-dimensional or lumped mass systems, e.g., \cite{Murray1996,Knueppel2010,Stroehle2022}.
The spatial Galerkin approximation of a geometrically exact string can generate a non-minimum phase system with many degrees of freedom, for which the mentioned approaches of inversion-based feed-forward control design remain elaborate and non-trivial tasks.

The intention of our contribution is to show the design of a feed-forward controller, which
\begin{itemize}
	\item can be applied for large scale structural mechanical (non-minimum phase) systems,
	\item does not require the explicit analysis of a semi-discrete system and its internal dynamics,
	\item can be easily implemented with commercial or open source finite element packages.
\end{itemize}
We achieve this by formulating an optimization problem to track the desired reference trajectory on a sufficiently large time interval (containing pre- and post-actuation), expressing the necessary stationarity conditions in variational form in space and time, and coding them in a finite element software. Using space-time finite elements, the optimal feed-forward controller is obtained from the solution of a set of algebraic equations.
Due to the special form of our variational formulation an explicit discretization of the velocities is not required.

The paper is organized as follows. In Section \ref{sec:2}, we recall the theory of continuum mechanics for geometrically exact strings and a corresponding Galerkin approximation with finite elements.
Afterwards, we define the overall problem and focus on the feed-forward control design in Section \ref{sec:3}.
This includes the variational formulation and its discretization.
In Section \ref{sec:4}, we test our approach with a numerical example and Section \ref{sec:5} concludes the contribution with an outlook.

\section{Geometrically exact strings}
\label{sec:2}
In this section, we recall the basic equations of geometrically exact strings and their spatial finite element approximation as a starting point for the formulation of our approach.
The interested reader is referred to \cite{Antman1995,Bonet2008} for more information.
\subsection{Problem description}
In the following, we consider the one-dimensional reference configuration ${\Omega_0= \left[ 0,L\right]\subset\mathbb{R}^d}$, $d=1,2,3$. The position of each particle in the reference configuration is described by the arc length coordinate $s=X_1\in\mathbb{R}$. The momentum balance of a string is given by
\begin{equation} \label{eq:balance}
	\rho(s) A(s)\partial^2_tr(s,t)=\partial_sn(s,t)+b(s,t) \quad \text{in} \; \Omega_0,
\end{equation}
where the position vector $r(s,t)$ describes the motion or actual configuration of the string, see Figure \ref{fig:configuration}. As usual $\rho(s)$ represents the density, $A(s)$ the cross-sectional area, $b(s,t)$ body forces, and $n(s,t)$ the normal force.

\begin{assum}
	In the present work, we use the simple nonlinear hyperelastic constitutive equation
	\begin{equation} \label{eq:constitutive}
		n(s,t)=EA(1-||\partial_sr(s,t)||^{-1})\partial_sr(s,t),
	\end{equation}
	where $E$ represents the Young's modulus. More precise constitutive laws are given in the literature, see \cite{Stroehle2022}.
\end{assum}

In order to complete the system description, we apply the boundary conditions
\begin{align} \label{eq:boundary}
	\begin{array}{ll}
		n(0,t) &= -u(t), \\
		n(L,t) &= 0,
	\end{array}
\end{align}
where $u(t)$ denotes the control input or actuating force vector.
This force should transfer the string from an initial configuration
\begin{align} \label{eq:initial}
	\begin{array}{ll}
		r(s,t_i) &= r_i(s), \\
		\partial_tr(s,t_i) &= v_i(s),
	\end{array}
\end{align}
to a final one
\begin{align} \label{eq:final}
	\begin{array}{ll}
		r(s,t_e) &= r_e(s), \\
		\partial_tr(s,t_e) &= v_e(s),
	\end{array}
\end{align}
whereby the end of the string
\begin{equation} 
	y(t)=r(L,t)
\end{equation}
should follow a desired trajectory $y_d(t)$ as closely as possible.
Before we continue with this task in Section \ref{sec:3}, we show how to set up the finite-dimensional state space model.

\begin{figure}[htbp]
	\centering
	\def\svgwidth{7.5cm}
	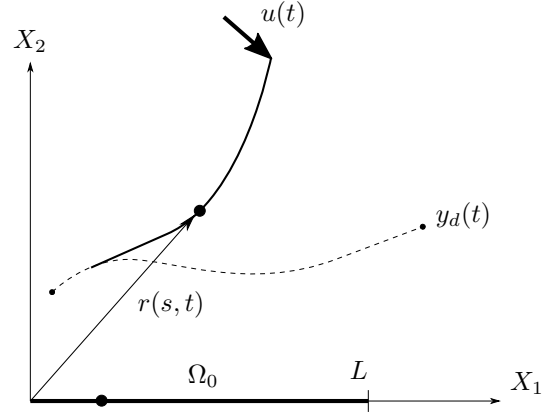
	\caption{Reference and actual configuration on $\mathbb{R}^2$.}
	\label{fig:configuration}
\end{figure}

\subsection{Weak formulation}
For ease of notation, we introduce 
\begin{align}
	\langle a,b \rangle_M=\int_{M} a \odot b \; \mathrm{d} s \notag
\end{align}
for the integration of products of scalar and vector valued functions over some domain $M$.
As a next step, we recall an appropriate weak form of equation \eqref{eq:balance}. 
By testing \eqref{eq:balance} and integrating by parts, one can see that any sufficiently regular solution of \eqref{eq:balance} under the boundary conditions \eqref{eq:boundary} is characterized by the variational single-field (displacement) formulation
\begin{align}
	\langle\rho A\partial^2_tr, w\rangle_{\Omega_0} &= -\langle  n, \partial_sw\rangle_{\Omega_0} + \langle b, w\rangle_{\Omega_0} + uw|_{s=0},
\end{align}
$\forall w\in H^1(\Omega_0)$, $t\geq0$.

\subsection{Spatial discretization}
In order to obtain a spatial approximation of above weak formulation, we have to find $r_h\in V_h\subset H^1(\Omega_0)$ solving the discretized weak form
\begin{align}
	\langle\rho A\partial^2_tr_h, w_h\rangle_{\Omega_0} &= -\langle  n, \partial_sw_h\rangle_{\Omega_0} + \langle b, w_h\rangle_{\Omega_0} + uw_h|_{s=0},
\end{align}
$\forall w_h\in V_h$, and $t\geq0$, where $V_h$ represents an appropriate finite dimensional subspace, e.g. a standard finite element space.
Accordingly, this discretization process leads to a system of ordinary differential equations of the form\footnote{The nonlinearity in ${\ttk(\ttr)}$ results from the nonlinear character of equation \eqref{eq:constitutive}.}
\begin{align}
	\label{eq:ode}
	\ttM\ddot{\ttr} &= -\ttk(\ttr) + \ttB + \ttG\ttu,
\end{align}
where $\ttr$ and $\dot{\ttr}$ denote the vector representations of the discretized functions.
Projecting the initial conditions \eqref{eq:initial} and final ones \eqref{eq:final} on these vector representations, gives us
\begin{align} \label{eq:boundary_conditions}
	\begin{array}{ll}
		\ttr(t_i) &= \ttr_i, \qquad
		\ttr(t_e) = \ttr_e, \\
		\dot{\ttr}(t_i) &= \ttv_i, \qquad
		\dot{\ttr}(t_e) = \ttv_e.
	\end{array}
\end{align}
In order to complete the state space representation \eqref{eq:ode} with regard to the above task, we use the corresponding output 
\begin{align} 
	\tty &= r_h(L,t) = \ttC\ttr.
\end{align}

\begin{rem}
	In order to solve an initial boundary value problem, an appropriate time-stepping scheme can be used to integrate the semi-discrete equation \eqref{eq:ode} in time.
	In the following, we will call this process a time-marching simulation.
	More information regarding this can be found in the Appendix.
\end{rem}

\section{Optimal trajectory feed-forward control design}
\label{sec:3}
In this section, we start with the cost functional to obtain an optimal feed-forward controller for the semi-discrete model, and cast it in terms of inner products on a space-time domain. We derive the corresponding variational identities, which, with  appropriate Galerkin approximation in time, can be immediately implemented as a space-time finite element problem in available software.

\begin{prob} \label{prob:1}
	Among all vectors $\ttr(t)$, $\ttu(t)$ that fulfill \eqref{eq:ode}
	and the boundary conditions \eqref{eq:boundary_conditions}
	find those that minimize the cost function
	\begin{align}
		\tilde{J} = \int_{t_i}^{t_e} \left( \frac{1}{2}\ttu^\top\ttu + \frac{\alpha}{2}(\ttC\ttr-\tty_d)^\top(\ttC\ttr-\tty_d)\right) \; \mathrm{d} t
	\end{align}
	with a factor $\alpha\in\mathbb{R}^{>0}$ penalizing the trajectory error compared to a quadratic input regularization term.
\end{prob}

\subsection{Lagrange functional}
We now consider a rectangular space-time domain $\Omega = \Omega_0\times T$, where ${T=[t_i,t_e]}$ is a sufficiently large time interval for possible phases of pre- and post-actuation, see Figure \ref{fig:space-time}.
To account for the system dynamics, we can define the Lagrange functional
\begin{align} \label{eq:cost_function}
	\begin{array}{ll}
		J &= \frac{1}{2}\langle u, u\rangle_{\Gamma_u} + \frac{\alpha}{2}\langle r_h-y_d, r_h-y_d\rangle_{\Gamma_y}  \\
		&\qquad - \langle \partial_sw_h, n \rangle_{\Omega} + \langle w_h, b -\rho A\partial^2_tr_h \rangle_{\Omega} + \langle w_h, u \rangle_{\Gamma_u}
	\end{array}
\end{align}
containing the Lagrange multiplier $w_h\in V_h(\Omega_0)\otimes H^1(T)$.

\begin{figure}[htbp]
	\vspace{0.3cm}
	\begin{center}
		\def\svgwidth{7.5cm}
		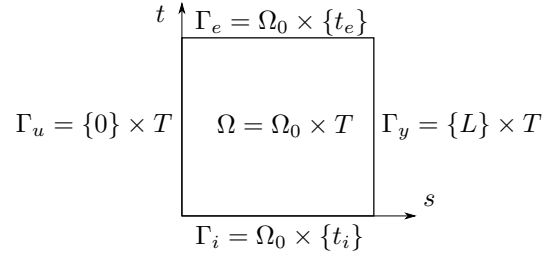
		\caption{Space-time domain.} 
		\label{fig:space-time}
	\end{center}
\end{figure}

\subsection{Variational identities}
As a next step, we derive the variational identities which will be approximated in the following section.
\begin{lem}
	Let $(r_h, u, w_h)$ be a sufficiently regular solution for Problem \ref{prob:1} minimizing \eqref{eq:cost_function}. Then the variational identities
	\begin{equation} \label{eq:var}
		\begin{array}{ll}
			\langle \delta u, u +w_h \rangle_{\Gamma_u}=0, \\[1ex]
			\langle \partial_t\delta w_h, \rho A\partial_tr_h \rangle_{\Omega} -\langle \partial_s\delta w_h, n \rangle_{\Omega} + \langle \delta w_h, b \rangle_{\Omega} \\
			\qquad + \langle \delta w_h, v_i \rangle_{\Gamma_i} - \langle \delta w_h, v_e \rangle_{\Gamma_e} + \langle \delta w_h, u \rangle_{\Gamma_u} = 0, \\[1ex]
			\langle \partial_t\delta r_h, \rho A \partial_t w_h \rangle_{\Omega} - \langle \partial_s\delta r_h, \nabla_{\partial_sr}n\cdot\partial_sw_h \rangle_{\Omega} \\
			\qquad + \alpha\langle \delta r_h, r_h - y_d \rangle_{\Gamma_y} = 0,
		\end{array}
	\end{equation}
	hold for all regular test functions $\delta r_h\in\mathcal{R}$, $\delta w_h\in\mathcal{W} $, and $\delta u\in\mathcal{U}$ with
	\begin{align}
		\mathcal{W} &= V_h(\Omega_0)\otimes H^1(T), \\
		\mathcal{R} &= \{\delta r_h\in \mathcal{W} \;|\; \delta r_h(s,t_i)=0, \; \delta r_h(s,t_e)=0\}, \\
		\mathcal{U} &= H^1(T).
	\end{align}
\end{lem}

\begin{proof}
	By considering the variational derivative of \eqref{eq:cost_function} and setting it to zero, one can see that
	\begin{equation}
		\begin{array}{ll}
			\delta J &=\langle \delta u, u +w_h \rangle_{\Gamma_u} \\
			&\qquad - \langle \partial_s\delta w_h, n \rangle_{\Omega} + \langle \delta w_h, b\rangle_{\Omega} -\underbrace{\langle \delta w_h,\rho A\partial^2_tr_h \rangle_{\Omega}}_{(*)}  \\
			&\qquad - \langle \partial_s\delta r_h, \nabla_{\partial_sr}n\cdot\partial_sw_h \rangle_{\Omega} - \underbrace{\langle \partial^2_t\delta r_h, \rho A w_h \rangle_{\Omega}}_{(**)} \\
			&\qquad + \alpha\langle \delta r_h, r_h - y_d \rangle_{\Gamma_y} + \langle \delta w_h, u \rangle_{\Gamma_u} = 0.
		\end{array}
	\end{equation}
	Integration by parts of $(*)$ leads to
	\begin{equation} \label{eq:ibp_1}
		\begin{array}{ll}
			(*) &= -\langle \partial_t\delta w_h, \rho A\partial_tr_h \rangle_{\Omega} \\ &\qquad - \langle \delta w_h, \rho Av_i \rangle_{\Gamma_i} + \langle \delta w_h, \rho Av_e \rangle_{\Gamma_e},
		\end{array}
	\end{equation}
	where the corresponding velocity boundary conditions are already substituted. Integration by parts of $(**)$ gives us
	\begin{equation} \label{eq:ibp_2}
		\begin{array}{ll}
			(**) &= -\langle \partial_t\delta r_h, \rho A\partial_tw_h \rangle_{\Omega} \\
			&\qquad - \langle \partial_t\delta r_h, \rho A w_h \rangle_{\Gamma_i} + \langle \partial_t\delta r_h, \rho w_h \rangle_{\Gamma_e}.
		\end{array}
	\end{equation}
	Due the fixed boundary conditions \eqref{eq:boundary_conditions} $\delta r_h$ can not be varied at the boundaries $\Gamma_i$, $\Gamma_e$, which makes the boundary terms in \eqref{eq:ibp_2} vanish.
\end{proof}

\subsection{Time discretization}
For the discretization in time, we likewise use a Galerkin approximation of the variational identities \eqref{eq:var}. As indicated in the introduction, we use standard finite element spaces for the approximation.

\begin{prop}
	By defining
	\begin{equation}
		\mathcal{W}_t \subset\mathcal{W},\; \mathcal{R}_t \subset\mathcal{R},\;\text{and}\; \mathcal{U}_t \subset\mathcal{U} \notag
	\end{equation}
	as appropriate finite dimensional subspaces, we can search for an approximation $(u_{t},w_{h_t}) \in \mathcal{U}_t \times \mathcal{W}_t$, and 
	\begin{equation}
		r_{h_t}\in\{r_{h_t}\in \mathcal{W}_t \;|\; r_{h_t}(s,t_i)= r_i,\; r_{h_t}(s,t_e)= r_e\}\notag
	\end{equation}
	for the solution of our problem satisfying the discretized identities
	\begin{equation} \label{eq:var_discrete}
		\begin{array}{ll}
			\langle \delta u_t, u_t +w_{h_t} \rangle_{\Gamma_u}=0, \\[1ex]
			\langle \partial_t\delta w_{h_t}, \rho A\partial_tr_{h_t} \rangle_{\Omega} -\langle \partial_s\delta w_{h_t}, n \rangle_{\Omega} + \langle \delta w_{h_t}, b \rangle_{\Omega} \\
			\qquad + \langle \delta w_{h_t}, v_i \rangle_{\Gamma_i} - \langle \delta w_{h_t}, v_e \rangle_{\Gamma_e} + \langle \delta w_{h_t}, u_t \rangle_{\Gamma_u} = 0, \\[1ex]
			\langle \partial_t\delta r_{h_t}, \rho A \partial_t w_{h_t} \rangle_{\Omega} - \langle \partial_s\delta r_{h_t}, \nabla_{\partial_sr}n\cdot\partial_sw_{h_t} \rangle_{\Omega} \\
			\qquad + \alpha\langle \delta r_{h_t}, r_{h_t} - y_d \rangle_{\Gamma_y} = 0,
		\end{array}
	\end{equation}
	for all $(\delta u_t, \delta r_{h_t}, \delta w_{h_t}) \in \mathcal{U}_t \times \mathcal{R}_t \times \mathcal{W}_t$.
\end{prop}

\begin{rem}
	Testing \eqref{eq:var_discrete} with ${\delta u_t=\delta w_{h_t}|_{\Gamma_u}\in\mathcal{W}_t|_{\Gamma_u}}$ allows us to eliminate the first line of \eqref{eq:var_discrete} and use ${\langle \delta w_{h_t}, u_t \rangle_{\Gamma_u}} = -{\langle \delta w_{h_t}, w_{h_t} \rangle_{\Gamma_u}}$ in the second line of \eqref{eq:var_discrete}.
	This reduces the amount of equations for the discretization process.
\end{rem}

Since we use bilinear basis functions of $s$ and $t$, our approximated computational domain represents a structured mesh with rectangular shaped space-time finite elements. The mesh is spanned by $n_s$ cells in $s$-direction and $n_t$ cells in $t$-direction, leading to ${n_s\times n_t}$ space-time finite elements.

\subsection{Discussion}
Our result from the space-time finite element perspective can be related to works based on semi-discrete models and the time integration with special time finite elements.
Since we assume that the space of $w_h$ consists of tensor products of bilinear functions
of $s$ and $t$, e.g. $w_h(s,t) = \phi(s)\ttw(t)$, the above Lagrange functional can be rewritten as
\begin{equation}
	\begin{array}{ll}
		J &= \frac{1}{2}\langle \ttu, \ttu\rangle_{T} + \frac{\alpha}{2}\langle \ttC\ttr-\tty, \ttC\ttr-\tty\rangle_{T} \\
		&\qquad + \langle \ttw, -\ttM\ddot{\ttr} -\ttk(\ttr) + \ttB + \ttG\ttu \rangle_{T},
	\end{array}
\end{equation}
which can also be found in the work of Altmann and Heiland \cite{Altmann2016}. Using the Euler-Lagrange equation, we can derive the optimality conditions
\begin{equation}
	\label{eq:optimality_conditions}
	\begin{array}{lll}
		\ttu &= - \ttG^\top\ttw, \\
		\ttM\ddot{\ttr} &= -\ttk(\ttr) + \ttB + \ttG\ttu, \\
		\ttM\ddot{\ttw} &= -\nabla_\ttr\ttk(\ttr)\ttw + (\ttC\ttr-\tty_d)\alpha.
	\end{array}
\end{equation}
In case of 
\begin{equation*}
	\mathcal{W}_t=[V_h(\Omega_0)\otimes P_1(T)]\cap H^1(\Omega)
\end{equation*}
the approximation of the variational identities corresponds to a multiple-shooting method \cite{Rao2010}, where the optimality conditions \eqref{eq:optimality_conditions} are integrated via a special (unbiased displacement-continuous, velocity-discontinuous) time finite element scheme \cite{Bauchau2024}.

\section{Numerical test}
\label{sec:4}
To demonstrate the performance of the above control design approach, we now provide some numerical results.

\subsection{Model scenario}
We consider a planar motion of a geometrically exact string under gravity. The suspended rope performs a finite-time transition between two stationary set-points, with the end of the string following the desired trajectory
\begin{equation}
	y_d(t)=\begin{bmatrix} 
		1 \\ 1
	\end{bmatrix}\psi(t)
\end{equation}
with
\begin{equation}
	\psi(t)=
	\begin{cases}
		0,\quad &t\leq \Delta t, \\
		3\left(\frac{t-\Delta t}{\Delta t}\right)^2 - 2\left(\frac{t-\Delta t}{\Delta t}\right)^3,\quad &\Delta t<t\leq 2\Delta t, \\
		1,\quad &t>2\Delta t. \\
	\end{cases}
\end{equation}

\begin{rem}
	Since \eqref{eq:balance} has a hyperbolic character the traveling time of information depends on the wave propagation speed. Therefore, ${y(t)}$ requires a pre- and post actuation  phase ${\Delta t}$ larger than the traveling time of a wave between ${s=0}$ and ${s=L}$ for steering the string from one steady-state to another one. This is especially relevant for flatness based control \cite{Knueppel2010}. A pre- and post-actuation phase in ${y_d(t)}$ allows us to approximate the inverse solution similar to \cite{Bastos2017}. Consequently, ${y(t)}$ follows ${y_d(t)}$ almost exactly, see also Figure \ref{fig:y}.
\end{rem}

The initial set point at time $t_i=0$ results from solving the equilibrium problem of \eqref{eq:ode} where ${V_h=P_1(\Omega_0)\cap H^1(\Omega_0)}$ and ${r(0,t)\equiv 0}$. Therefore, $n_s=10$ equidistant finite elements were used. The final set-point at time $t_e=6$ corresponds to the initial one except that the positional arguments are shifted by one in both directions. The remaining parameters of the numerical experiment are given in Table \ref{tab:parameter}.

\begin{table}[htbp]
	\caption{Parameters}
	\label{tab:parameter}
	\centering \medskip
	{\small
		\begin{tabular}{ l l l l l l}
			\hline
			$L$ & 
			$\rho$ & 
			$EA$ & 
			$g$ &
			$\Delta t$ &
			$\alpha$ \\
			\hline
			$1$ & $1$ & $1$ & $9.81$ & 2& 100 \\
			\hline
		\end{tabular}
	}
\end{table}

\subsection{Details on the computational approximation process}
Since we use a finite element method for the space-time discretization, the computational domain is partitioned into ${10\times 100}$ uniform rectangles. 
The trial functions of \eqref{eq:var_discrete} are generated via continuous Lagrangian shape functions of first order what leads to the approximation spaces
\begin{equation}
	\mathcal{W}_t=[P_1(\Omega_0)\otimes P_1(T)]\cap H^1(\Omega), \qquad \mathcal{U}_t=\mathcal{W}_t|_{\Gamma_u}. \notag
\end{equation}
The following computational results are generated with the finite element framework FEniCSx \cite{Baratta2023,Scroggs2022a,Scroggs2022b,Alnaes2014}.

\subsection{Numerical results}
Solving \eqref{eq:var_discrete} according to the above details, we generate several illustrations.
E.g., Figure \ref{fig:r_x} shows the position field $r(s,t)$ in $X_1$-direction and the finite element mesh. The components of the corresponding input force vector $u(t)$ are given in Figure \ref{fig:u}.
To verify our results we apply the generated input $u(t)$ to a time-marching simulation. It uses an implicit midpoint rule for the time integration and the time step size $\tau=0.06$. Further information on the time-marching simulation is given in the Appendix. In Figure \ref{fig:y}, the resulting output $y(t)$ in $X_1$-direction is presented\footnote{Figure \ref{fig:u} and \ref{fig:y} also contain the solutions for lower values of $\alpha$ to illustrate the increased deviation from the desired trajectory if it is less penalized in the quadratic cost functional.}.
In order get an impression of the overall motion, Figure \ref{fig:snapshots} shows several snapshots of the string configuration which also result from the time-marching simulation.

\begin{figure}[htbp]
	\centering
	\includegraphics[width=0.5\textwidth]{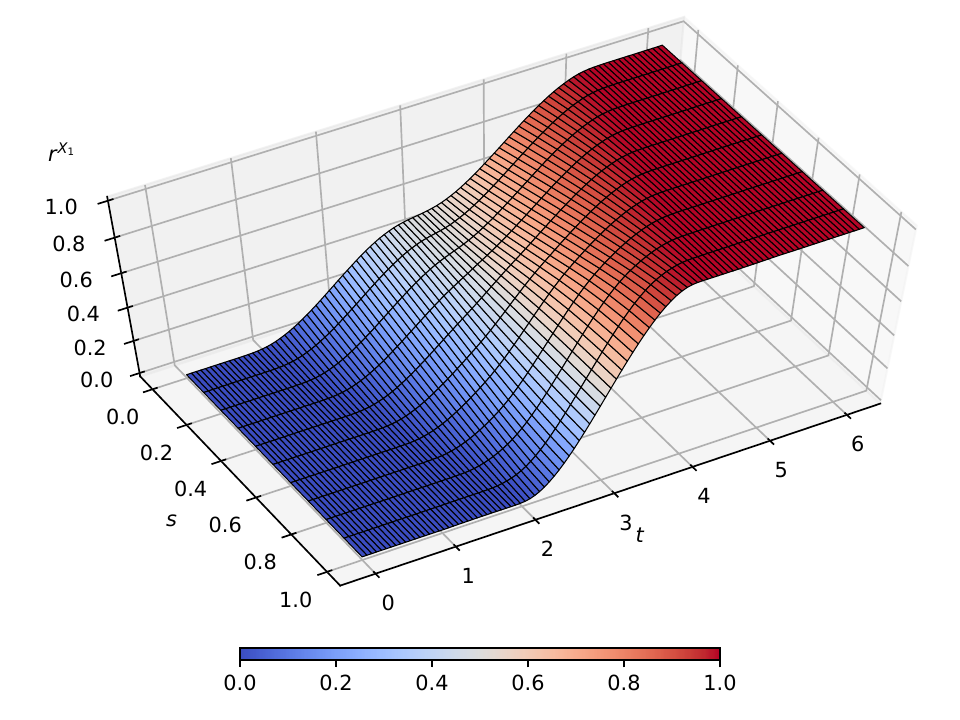}
	\caption{Numerical solution of $r(s,t)$ in $X_1$-direction.}
	\label{fig:r_x}
\end{figure}

\begin{figure}[htbp]
	\centering
	\includegraphics[width=0.5\textwidth]{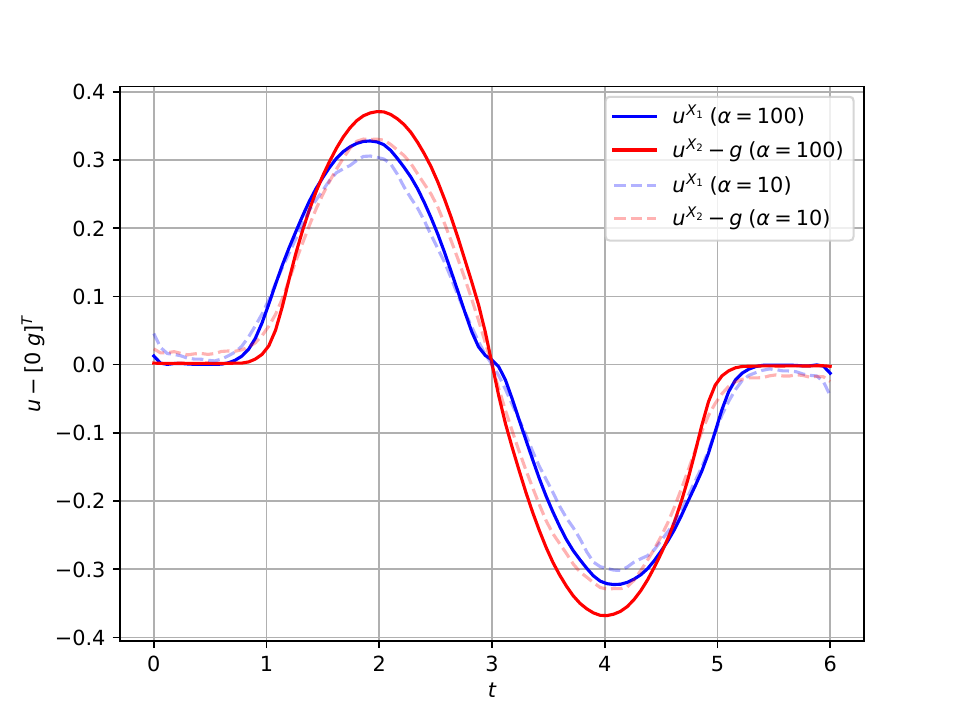}
	\caption{Numerical solution of $u(t)$. In order to present the input force in $X_1$-direction and the one in $X_2$-direction in the same diagram, we subtract the gravitational portion.}
	\label{fig:u}
\end{figure}

\begin{figure}[htbp]
	\centering
	\includegraphics[width=0.5\textwidth]{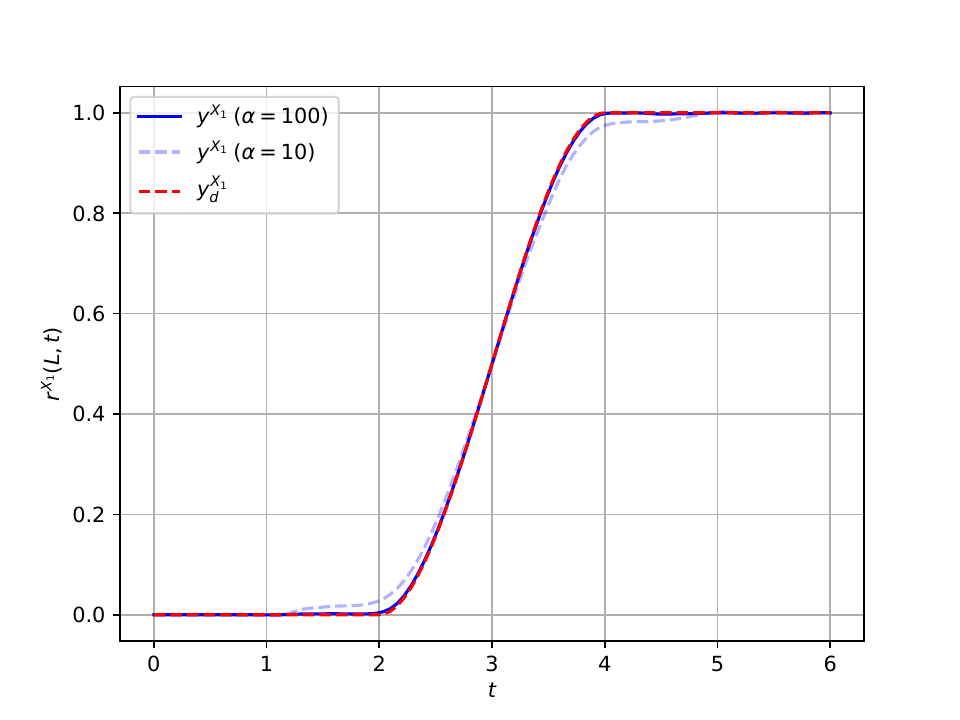}
	\caption{Actual output $y(t)$ compared with $y_d(t)$. In the case of a large $\alpha$, $y(t)$ is close to $y_d(t)$ . The results in $X_2$-direction are similar and omitted for the sake of clarity.}
	\label{fig:y}
\end{figure}

\begin{figure}[htbp]
	\centering
	\includegraphics[width=0.5\textwidth]{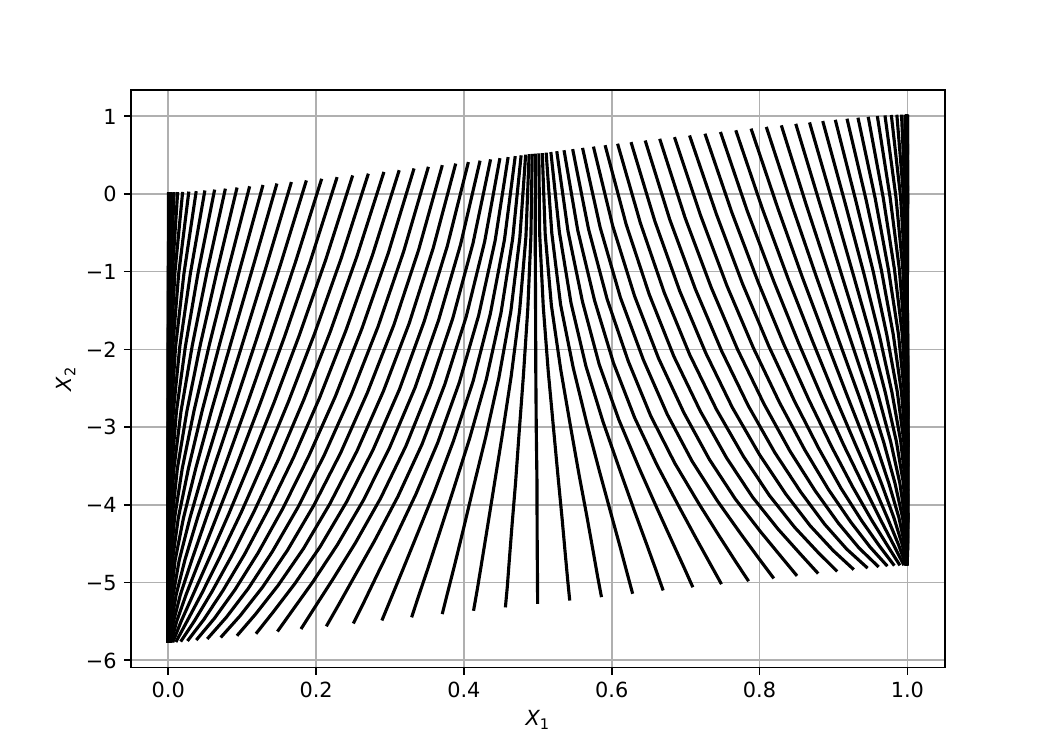}
	\caption{Snapshots at different time steps starting from the left with the initial string configuration at time ${t=t_i}$ and ending in the final one (${t=t_e}$) on the right.}
	\label{fig:snapshots}
\end{figure}

\section{Conclusion}
\label{sec:5}
We presented how to generate an optimal trajectory feed-forward controller for geometrically exact strings exploiting a space-time variational formulation of the problem. The big advantage of this approach, which is related to corresponding formulations for semi-discretized models and the solution of their two-point boundary value problems, is the ease of implementation in commercial or free finite element software, in conjunction with the efficient solution of the resulting algebraic system of equations.

In future research, we might extend the control design to three-dimensional systems. This goes hand in hand with a more complex implementation if the favored finite element toolbox is not compatible with four-dimensional meshes.
Moreover, the challenges or advantages of unstructured space-time finite element meshes or other Galerkin approximation methods have not been discussed yet.
Additionally, we want to extend the feed-forward control design to other cost functionals.
Consequently, and depending on its real-time capability, the presented control design might also be relevant for model predictive control.
                                  
%%%%%%%%%%%%%%%%%%%%%%%%%%%%%%%%%%%%%%%%%%%%%%%%%%%%%%%%%%%%%%%%%%%%%%%%%%%%%%%%
\section*{APPENDIX}
\label{app:a}
In this contribution, the numerical time integration of \eqref{eq:ode} is called a time-marching simulation; e.g., in Section \ref{sec:4} we applied an implicit midpoint rule.
Therefore, we introduce the velocity $\ttv=\dot{\ttr}$ to rewrite \eqref{eq:ode} as a first order system ${\dot{\ttx}=\ttf(\ttx,\ttu)}$ with ${\ttx^\top=[\ttr^\top\;\ttv^\top]}$;
and define
\begin{equation}
	d_\tau\ttx=\tau^{-1}(\ttx^{n+1}-\ttx^n),\quad\text{and}\quad \ttx^{n+\frac{1}{2}}=2^{-1}(\ttx^{n+1}+\ttx^n),\notag
\end{equation}
where $\tau$ is a fixed time step size and $t^n=n\tau$, ${n\geq 0}$ the discrete time step. We now use
\begin{equation}
	d_\tau\ttx=\ttf(\ttx^{n+\frac{1}{2}}, \ttu^{n+\frac{1}{2}}) \notag
\end{equation}
to obtain $\ttx^{n+1}\approx\ttx(t^{n+1})$ at the next time step $t^{n+1}$.

\section*{ACKNOWLEDGMENT}
The authors would like to thank Tengman Wang, Levent \"Ogretmen, and Philipp L. Kinon for insightful discussions.
The support by the DFH-UFA French-German doctoral college "Port-Hamiltonian Systems: Modeling, Numerics and Control" is gratefully acknowledged.

%%%%%%%%%%%%%%%%%%%%%%%%%%%%%%%%%%%%%%%%%%%%%%%%%%%%%%%%%%%%%%%%%%%%%%%%%%%%%%%%
\bibliographystyle{IEEEtran}
\bibliography{ifacconf}

\end{document}